%% file: main.tex
\documentclass[fullpage,11pt]{article}

\usepackage{amsthm}
\usepackage{graphicx} 
\usepackage{array} 

\usepackage{amsmath, amssymb, amsfonts, verbatim}
\usepackage{hyphenat,epsfig,subfigure,multirow}
\usepackage{dsfont, tikz}

\tikzset{%
  zeroarrow/.style = {-stealth,dashed},
  onearrow/.style = {-stealth,solid},
  c/.style = {circle,draw,solid,minimum width=2em,
        minimum height=2em},
  r/.style = {rectangle,draw,solid,minimum width=2em,
        minimum height=2em}
}

\usepackage[lined,boxed,ruled,norelsize,algo2e,linesnumbered]{algorithm2e}

\usepackage{algorithm,algorithmic}

\usepackage{tcolorbox}

\definecolor{DarkGreen}{rgb}{0.1,0.5,0.1} 
\definecolor{DarkRed}{rgb}{0.5,0.1,0.1}
\definecolor{DarkBlue}{rgb}{0.1,0.1,0.5}

\usepackage{hyperref}
\hypersetup{
colorlinks=true,
pdfnewwindow=true,
citecolor=DarkGreen,
linkcolor=DarkRed,
filecolor=DarkRed,
urlcolor=DarkBlue
}

\usepackage{bm}
\usepackage{url}
\usepackage{xspace} 
\usepackage[mathscr]{euscript}

\usepackage{mdframed}


\usepackage{cite}
\usepackage{enumerate}

\usepackage[margin=1in]{geometry}

\newtheorem{theorem}{Theorem}
\newtheorem{lemma}{Lemma}[section]
\newtheorem{proposition}[lemma]{Proposition}
\newtheorem{corollary}[theorem]{Corollary}
\newtheorem{claim}[lemma]{Claim}

\newtheorem{definition}{Definition}

\newtheorem{problem}{Problem}
\newtheorem{remark}[lemma]{Remark}

\newtheorem{mdresult}{Result}

\newtheorem{mdinvariant}{Invariant}

\newtheorem*{theorem*}{Theorem}
\newtheorem*{claim*}{Claim}
\newtheorem*{proposition*}{Proposition}
\newtheorem*{lemma*}{Lemma}
\newtheorem*{problem*}{Problem}

\allowdisplaybreaks

\usepackage{thmtools}
\usepackage{thm-restate}

\DeclareMathOperator*{\argmax}{arg\,max}

\renewcommand{\qed}{\nobreak \ifvmode \relax \else
      \ifdim\lastskip<1.5em \hskip-\lastskip
      \hskip1.5em plus0em minus0.5em \fi \nobreak
      \vrule height0.75em width0.5em depth0.25em\fi}

\usepackage[T1]{fontenc}
\usepackage[utf8]{inputenc}

\input{macros}

\renewcommand{\O}{\mathcal{O}}

\title{Secretary Ranking with Minimal Inversions\footnote{Work done in part while the first two authors were summer interns at Google Research, New York.}}
\author{Sepehr Assadi\thanks{Department of Computer and Information Science, University of Pennsylvania. Email: \texttt{sassadi@cis.upenn.edu}.} \and 
Eric Balkanski\thanks{School of Engineering and Applied Sciences, Harvard University. Email: \texttt{ericbalkanski@g.harvard.edu}.}  \and  
Renato Paes Leme\thanks{Google Research, New York. Email: \texttt{renatoppl@google.com}.}
}

\date{}

\begin{document}

\maketitle
\thispagestyle{empty}
\begin{abstract}
We study a twist on the classic secretary problem, which we term the \emph{secretary ranking} problem: elements from an ordered set
arrive in random order and instead of picking the maximum element, the algorithm is asked to
assign a rank, or position, to each of the elements. The rank assigned is irrevocable and is
given knowing only the pairwise comparisons with elements previously arrived.
The goal is to minimize the distance of the rank produced to the true rank of
the elements measured by the Kendall-Tau distance, which corresponds to the
number of pairs that are inverted with respect to the true order.

\medskip

Our main result is a matching upper and lower bound for the secretary ranking
problem. We present an algorithm that ranks $n$ elements with only
$O(n^{3/2})$ inversions in expectation, and show that any algorithm
necessarily suffers $\Omega(n^{3/2})$ inversions when there are $n$ available positions.  In terms of techniques, the
analysis of our algorithm draws connections to linear probing in the hashing
literature, while our lower bound result relies on a general
anti-concentration bound for a generic balls and bins sampling process.
We also consider the case where the number of positions $m$
can be larger than the number of secretaries $n$ and provide an improved
bound by showing a  connection of this problem with random binary trees.
\end{abstract}
\setcounter{page}{0}

\newpage


\input{intro}


\input{alg}

\input{lower-bound}

\input{sparse}

\subsection*{Acknowledgements} 
We thank Sanjeev Khanna, Robert Kleinberg, and Vahab Mirrokni for helpful comments. 

\bibliographystyle{abbrv}
\bibliography{general}

\appendix

\input{prelim} 
\input{appendixAlg}
\input{appendix}

\end{document}

%% file: macros.tex
\newcommand{\toShrink}{-.20cm}
\newcommand{\toShrinkEnu}{-.2cm}

\newcommand{\Leq}[1]{\ensuremath{\underset{\textnormal{#1}}\leq}}

\newenvironment{tbox}{\begin{tcolorbox}[
		enlarge top by=5pt,
		enlarge bottom by=5pt,
		 boxsep=0pt,
                  left=4pt,
                  right=4pt,
                  top=10pt,
                  arc=0pt,
                  boxrule=1pt,toprule=1pt,
                  colback=white
                  ]
	}
{\end{tcolorbox}}

\newcommand{\eps}{\ensuremath{\varepsilon}}
\newcommand{\Paren}[1]{\Big(#1\Big)}

\newcommand{\bracket}[1]{\left[#1\right]}
\newcommand{\paren}[1]{\ensuremath{\left(#1\right)}\xspace}
\newcommand{\card}[1]{\left\vert{#1}\right\vert}

\newcommand{\ceil}[1]{{\left\lceil{#1}\right\rceil}}

\newcommand{\alg}{\ensuremath{\mathcal{A}}\xspace}

\DeclareMathOperator*{\Exp}{\ensuremath{{\mathbb{E}}}}
\DeclareMathOperator*{\Prob}{\ensuremath{\textnormal{Pr}}}
\renewcommand{\Pr}{\Prob}

\newcommand{\Ex}{\Exp}

\newcommand{\etal}{et al.\xspace}

\renewcommand{\alg}{\ensuremath{\textnormal{\textsf{ALG}}}\xspace}

\newcommand{\rank}[1]{\ensuremath{\textnormal{\textsf{rk}}(#1)}\xspace}

\newcommand{\erank}[1]{\ensuremath{\widetilde{\textnormal{\textsf{rk}}}(#1)}\xspace}

\newcommand{\A}{\mathcal{A}}

%% file: intro.tex

\newcommand{\II}{\ensuremath{\mathbb{I}}}

\section{Introduction}
\label{sec:intro}

The secretary problem is one of the first problems studied in online algorithms---in fact, is was extensively studied much before the field of online
algorithms even existed. It first appeared in print in 1960 as a recreational
problem in Martin Gardner's Mathematical Games column in  Scientific American.
In the subsequent decade it caught the attention of many of the eminent
probabilist researchers like Lindley~\cite{lindley1961dynamic}, Dynkin~\cite{dynkin1963optimum}, Chow \etal~\cite{chow1964optimal} and Gilbert and Mosteller \cite{gilbert2006recognizing}
among others. In a very entertaining historical survey,
Ferguson \cite{ferguson1989solved} traces the origin of the secretary problem to
much earlier: Cayley in 1875 and Kepler in 1613 pose questions in the same spirit
as the secretary problem.

Secretary problem has been extended in numerous directions, see for example the
surveys by Sakaguchi \cite{sakaguchi1995optimal} and Freeman \cite{freeman1983secretary}.
The problem has had an enormous influence in computer science and has provided
some of basic techniques in the field of online and approximation algorithms.
Babaioff et al extended this problem to matroid set systems
\cite{babaioff2007matroids} and Knapsack \cite{babaioff2007knapsack} and perhaps more
importantly, show that the secretary problem is a natural tool for designing
online auctions. In the last decade, the secretary problem has also been
extended to posets \cite{kumar2011hiring}, submodular systems
\cite{bateni2010submodular},
general set systems
\cite{rubinstein2016beyond}, stable matchings \cite{babichenko2017stable},
non-uniform arrivals \cite{kesselheim2015secretary}
and applied to optimal data sampling \cite{girdhar2009optimal},
design of prophet inequalities \cite{azar2014prophet,esfandiari2017prophet},
crowdsourcing systems \cite{singer2013pricing}, pricing in online
settings \cite{cohen2014pricing}, online linear programming
\cite{agrawal2014dynamic} and online ad
allocation \cite{feldman2010online}.

The (admittedly incomplete) list of extensions and applications in the last
paragraph serves to showcase that the secretary problem has traditionally 
been a vehicle for deriving connections between different subfields of computer
science and a testbed of new techniques.

\paragraph{Ranking Secretaries.}
Here we consider a natural variant of the secretary problem, which we name the \emph{secretary ranking} problem, where instead of
selecting the maximum element we are asked to \emph{rank} each arriving element. In the
process of deriving the optimal algorithm for this problem, we show connections
to the technique of linear probing, which is one of the earliest techniques in
the hashing literature studied by Knuth~\cite{knuth1963notes} and also to
the expected height of random binary trees.

In the traditional secretary problem a decision maker is trying to hire a
secretary. There is a total order over $n$ secretaries and the goal of the algorithm
is to hire the best secretary. The secretaries are assumed to arrive in a random order and the
algorithm can only observe the relative rank of each secretary with respect to
the previously interviewed ones. Once a secretary is interviewed, the algorithms
needs to decide whether to hire the current one or to irrevocably abandon the
current candidate and continue interviewing.

In our setting, there are $m$ job positions and $n$ secretaries 
with $m \geq n$. There is a known total order on positions.
Secretaries arrive in random
order and, as before, we can only compare a secretary with previously
interviewed ones. In our version, all secretaries will be hired and the decision
of the algorithm is in which position to hire each secretary. Each position can
be occupied by at most one secretary and hiring decisions are irrevocable.
Ideally, the algorithm will hire the best secretary in the best position,
the second best secretary in the second best position and so on. The loss
incurred by the algorithm corresponds to the pairs that are incorrectly
ordered, i.e., pairs where a better secretary is hired in a worse position.

\subsection{Our Results and Techniques}

\paragraph{Dense case ($m=n$).} We start by studying the perhaps most natural case of the secretary ranking problem when $m=n$, which we call the dense case.
The trivial algorithm that assigns a random empty position for each arriving
secretary incurs $\Theta(n^2)$ cost, since each pair of elements has probability
$1/2$ of being an inversion.
On the other hand, $\Omega(n)$ is  a trivial lower bound on the cost of any
algorithm  because nothing is known when the first element arrives. As such,
there is a linear gap between the costs of the trivial upper and lower bounds
for this secretary ranking problem. 
Our first main result is an asymptotically tight upper and lower bound on the loss incurred by the algorithms for the secretary ranking problem.

\begin{theorem*}\label{res:main}
  There is an algorithm for the secretary ranking
  problem with $m=n$ that computes a ranking with $\mathcal{O}(n^{3/2})$ inversions in
  expectation. Moreover, any algorithm for this problem makes $\Omega(n^{3/2})$
  inversions in expectation.
\end{theorem*}

There are two main challenges in designing an algorithm for secretary ranking. In
earlier time steps, there are only a small number of comparisons observed and
these do not contain  sufficient information to estimate the true rank of the
arriving elements.   In later time steps,  we observe a large number  of
comparisons and using the randomness of elements arrival, the true rank of the
elements can be estimated well. However, we face another difficulty for these
time steps:  many of the positions have already been assigned to some element
arrived earlier and are hence not available. The first information-theoretic
challenge exacerbates this second issue.  Previous bad  placements might imply
that all the desired positions are unavailable for the current element, causing
a large cost even for an element which we can estimate its true rank accurately.  

The algorithm needs to handle these two opposing challenges simultaneously. The
main idea behind our algorithm is to estimate the rank of the current element
using the observed comparisons and  then add some noise to these estimations to
obtain additional randomness in the positions and avoid positively correlated
mistakes. We then  assign the current element to the closest empty position to this noisy estimated rank. The main technical interest is in the
analysis of this algorithm. We draw a connection to the analysis of linear
probing in the hashing literature~\cite{knuth1963notes} to argue that under
this extra noise,  there exists typically an empty position that is close to the
estimated rank.

For the lower bound, we analyze the number of random pairwise comparisons needed
to estimate the rank of an element accurately. Such results are typically
proven in the literature by using {anti-concentration} inequalities. A main
technical difficulty is that most of the existing anti-concentration
inequalities are for {independent} random variables while there is a
{correlation} between the variables we are considering. We  prove, to the
best of our knowledge, a new anti-concentration inequality for a generic balls
in bins problem that involves correlated sampling.

\paragraph{Sparse Case ($m \gg n$).}
Next we consider the case where there are considerably more positions than secretaries and
compute how large the number $m$ of positions needs to be such that we incur no
inversions. Clearly for $m = 2^{n+1}-1$ it is possible to obtain zero inversions
with probability $1$ and for any number less than that it is also clear that any
algorithm needs to cause inversions with non-zero probability. If we only want
to achieve zero inversions with high probability, how large does $m$ need to be? 
By showing a connection between the secretary problem and random binary trees,
we show that for  $m \geq n^\alpha$ for $\alpha \approx 2.998$ it is possible
to design an algorithm that achieves zero inversion with probability $1 - o(1)$.
The constant $\alpha$ here is obtained using the high probability bound on the height a random binary tree of $n$ elements.

\paragraph{General Case.}
Finally, we combine the algorithms for the dense and sparse cases to obtain a general algorithm with a bound on the expected number of inversions which smoothly interpolates between the bounds obtained for the dense and sparse cases. 
This algorithm starts by running the algorithm for the sparse case and when two
elements are placed very closed to each other by the sparse algorithm, we switch
to use the algorithm for the dense case to assign a position to remaining elements with rank between these two close elements.

\subsection{Related Work}

Our work is inserted in the vast line of literature on the secretary problem,
which we briefly discussed earlier. There has been a considerable amount of work
on  multiple-choice secretary problems where, instead of the single best
element,  multiple elements can be chosen as they arrive online
\cite{kleinberg2005multiple,babaioff2007knapsack,babaioff2007matroids,bateni2010submodular,rubinstein2016beyond,KorulaP09}.
We note that in multiple-choice secretary problems, the decision at arrival of
an element is still binary, whereas in secretary ranking one of $n$ positions must
be chosen. More closely related to our work is a paper of Babichenko \etal~\cite{babichenko2017stable} where
elements that arrive must also be assigned to a position. However, the objective
is different and the goal, which uses a game-theoretic notion of stable
matching, is to maximize the number of elements that are not in a blocking pair. 
Gobel~\etal~\cite{GobelKT15} also studied an online appointment scheduling problem in which the goal is to assign starting dates 
to a set of jobs arriving online. The objective here is again different from the secretary ranking problem and is to minimize the total weight time of the jobs. 

Another related line of work in machine learning is the
well-known problem of learning to rank that has been extensively studied in recent years (e.g.
\cite{burges2005learning,cao2007learning,burges2007learning,xia2008listwise}).
Two important applications of this problem are  search engines for document
retrieval
\cite{liu2009learning,radlinski2005query,liu2007letor,cao2006adapting,xu2007adarank}
and collaborative filtering approaches to recommender systems
\cite{shi2010list,shi2012climf,liu2008eigenrank,wang2008probabilistic}. There
has been significant  interest recently in ranking from pairwise comparisons
\cite{feige1994computing,busa2013top,chen2015spectral,shah2015simple,jang2016top,
heckel2016active,davidson2014top,braverman2016parallel,agarwal2017learning}. To the best of our
knowledge, there has not been previous work on ranking from pairwise comparisons
in the online setting.

Finally, we also briefly discuss hashing, since our main technique is related to
linear probing. Linear probing is a
classic implementation of hash tables and was first analyzed theoretically by
Knuth in 1963~\cite{knuth1963notes}, in a report which is now regarded as the
birth of algorithm analysis. Since then, different variants of this problem
mainly for hash functions with limited independence have been considered in the
literature~\cite{schmidt1990analysis,pagh2007linear,patracscu2010k}. Reviewing
the vast literature on this subject is beyond the scope of our paper and we
refer the interested reader to these papers for more details. 

\smallskip
\noindent
\textbf{Organization.} The rest of the paper is organized as follows.
In Section~\ref{sec:setup} we formalize the problem and present the notation
used in the paper. In Section~\ref{sec:alg}, we present and analyze our
algorithm for $m=n$ case.
Section~\ref{sec:lower-bound} is devoted to showing the lower bound also for
$m=n$ case. Our results for the general case when $m$ can be different from $n$ appear in Section~\ref{sec:sparse}.
Missing proofs and standard concentration bounds are
postponed to the appendix.

\section{Problem Setup}\label{sec:setup}

In the secretary ranking problem, there are $n$ elements $a_1, \hdots, a_n$ that
arrive one at a time in an online manner and in a uniformly random order. There is a total ordering among the
elements, but  the algorithm has only access to pairwise
comparisons among the elements that have already arrived. In other words, at
time $t$, the algorithm only observes whether $a_i < a_j$ for all $i,j \leq t$. 

We define the rank function $\textsf{rk} : \{a_1, \hdots, a_n\} \rightarrow [n]$ as the
true rank of the elements in the total order, i.e., $a_i < a_j$ iff $\rank{a_i}
< \rank{a_j}$. 
Since the elements arrive uniformly at random, $\rank{\cdot}$ is a
random permutation. Upon arrival of an element $a_t$ at time step $t$, the algorithm must, irrevocably,
place $a_t$ in a position $\pi(a_t) \in [m]$ that is not yet occupied, in the
sense that for $a_t \neq a_s$ we must have $\pi(a_s) \neq \pi(a_t)$. Since the
main goal of the algorithm is to place the elements as to reflect the true rank
as close as possible\footnote{In other words, hire the better secretaries in
better positions}, we refer to $\pi(a_t)$ as the \emph{learned rank} of $a_t$. The goal is to
minimize the number of pairwise mistakes induced by the learned ranking compared
to the true ranking. A pairwise mistake, or an inversion, is defined as a pair
of elements $a_i, a_j$ such that $\rank{a_i} < \rank{a_j}$ according to the true
underlying ranking but $\pi(a_i) > \pi(a_j)$ according to the learned ranking.

The secretary ranking problem generalizes the secretary problem in the following
sense: in the secretary problem, we are only interested in finding the element
with the highest rank. 
However, in the secretary ranking problem, the goal is to assign a rank to every
arrived element and construct a complete ranking of all elements. Similar 
to the secretary problem, we make the enabling assumption  that  the order of
elements arrival is  uniformly random.\footnote{It is straightforward to verify
that when the ordering is adversarial, any algorithm incurs the trivial cost of
$\Omega(n^2)$ for $m=n$. For completeness, a proof is provided in
Appendix~\ref{sec:appadversarial}.} We measure the cost of the algorithm in
expectation over the randomness of both the arrival order of elements and the
algorithm.

\paragraph{Measures of sortedness.}
We point out that the primary goal in the secretary ranking problem is to learn an
ordering $\pi$ of the input elements which is as close as possible to their
sorted order. As such, the \emph{cost} suffered by an algorithm is given by a
\emph{measure of sortedness} of $\pi$ compared to the true ranking. There are
various measures of sortedness studied in the literature depending on the
application. Our choice of using the number of inversions, also known as
\emph{Kendall's tau} measure, as the cost of algorithm is motivated by the
importance of this measure and its close connection to other measures such as
\emph{Spearman's footrule} (see, e.g., Chapter 6B in~\cite{Diaconis88}). 

For a mapping $\pi:[n] \rightarrow [m]$, Kendall's tau $K(\pi)$ measures the number of
inversions in $\pi$, i.e.:
$$K(\pi): = \vert\{(i,j); (\pi(a_i) - \pi(a_j))(\rank{a_i} - \rank{a_j}) < 0
\}\vert.$$
When $n=m$, another important measure of sortedness is Spearman's footrule $F(\pi)$ given by:
$ F(\pi) :=
\sum_{i=1}^{n} \card{\rank{a_i}- \pi(a_i)},$ which corresponds to the summation of distances
between the true rank of each element and its current position.  A celebrated
result of Diaconis and Graham~\cite{DiaconisG77} shows that these two measures
are within a factor of two of each other, namely, $K(\pi) \leq F(\pi)
\leq 2 \cdot K(\pi)$. We refer to this inequality as the DG inequality throughout the paper.
Thus, up to a factor of two, the goals of minimizing the Kendall tau or
Spearman's footrule distances are equivalent and, while the Kendall tau distance
is used in the formulation of the problem, we also use the Spearman's footrule
distance in the analysis.

%% file: alg.tex

\newcommand{\set}[1]{\ensuremath{\{#1\}}}
\section{Dense Secretary Ranking}\label{sec:alg}

We start by analyzing the dense case, where $m = n$ and both the true
rank $\rank{\cdot}$ and the learned rank $\pi(\cdot)$ are permutations over $n$
elements. Our main algorithmic result is the following theorem.

\begin{theorem}\label{thm:dense}
There exists an algorithm for the secretary ranking problem with $m=n$ that incurs a cost of $O(n\sqrt{n})$ in expectation.
\end{theorem}

\subsection{Description of the Algorithm}

The general approach behind the algorithm in Theorem~\ref{thm:dense} is as follows.
\smallskip
\begin{mdframed}[hidealllines=false,backgroundcolor=gray!10,innertopmargin=10pt]
	Upon the arrival of element $a_{t}$ at time step $t$:
	\vspace{-2pt}
\begin{enumerate}
	\item \textbf{Estimation step:} Estimate the true rank of the arrived element $a_{t}$ using the \emph{partial} comparisons seen so far.
	\item \textbf{Assignment step:} Find the nearest currently unassigned rank to this estimate and let $\pi(a_{t})$ be this position.
\end{enumerate}
\end{mdframed}

We now describe the algorithm in more details. A natural way to estimate the rank of the $t$-th element in the estimation step is to compute
the rank of this element with respect to the previous $t-1$ elements seen so far and then scale this number to
obtain an estimate to rank of this element between $1$ and $n$. However, for our analysis of the assignment step, we need to tweak this approach slightly: instead of simply rescaling and rounding, we add  perturbation
to the estimated rank and then round its value. This gives a nice distribution of estimated ranks which is crucial for the analysis of the assignment step. The assignment step then simply 
assigns a learned rank to the element as close as possible to its estimated rank. We formalize the algorithm in Algorithm~\ref{alg:dense}.

\medskip

\begin{algorithm2e}[H]

\caption{Dense Ranking}\label{alg:dense}

{ \textbf{Input:} a set of $n$ positions, denoted here by $[n]$, and at most $n$ online arrivals.}

\For{\textnormal{any time step $t \in [n]$ and element $a_{t}$}}{

	Define $r_t := \card{\set{a_{t'} \mid a_{t'} < a_{t} \textnormal{~and~} t' < t}}$.

  Sample $x_t$ uniformly in the real interval $[r_t \cdot
  \frac{n}{t}, (r_t+1) \cdot \frac{n}{t}]$ and choose $\erank{a_{t}} = \lceil
  x_t \rceil$.
	

	Set the learned rank of $a_{t}$ as $\pi(a_{t}) = \arg\min_{i \in R}
  \card{i - \erank{a_{t}}}$ and remove $i$ from $R$.
	
}
\end{algorithm2e}

\medskip

We briefly comment on the runtime of the algorithm. By using any self-balancing binary search tree---such as a red-black tree or an AVL tree---to store the ranking of the arrived elements as well as the set $R$ of available ranks separately, 
Algorithm~\ref{alg:dense} is implementable in $O(\log n)$ time for each step, so total $O(n\log{n})$ worst-case time.

We also note some similarity between this algorithm and linear probing in hashing. Linear probing is an approach to resolving collisions in hashing where, when a  key is hashed to a non-empty cell,  the closest neighboring cells are visited until  an empty location is found for the key. The similarity is apparent to our assignment step which finds the nearest currently unassigned rank to the estimated rank of an element. The analysis of the assignment step follows similar ideas as the analysis for the linear probing hashing scheme.

\subsection{The Analysis}\label{sec:analysis} 

For $m=n$, the total number of inversions can be approximated within a factor
of $2$ by the Spearman's footrule. Therefore, we can write the cost of
Algorithm~\ref{alg:dense} (up to a factor $2$) as follows:
$$	\sum_{t=1}^{n} \card{\rank{a_{t}} - \pi(a_{t})} \leq \sum_{t=1}^{n} \card{\rank{a_{t}} - \erank{a_{t}}} + \sum_{t=1}^{n} \card{\erank{a_{t}} - \pi(a_{t})}. \label{eq:cost}$$

This basically breaks the cost of the algorithm in two parts: one is the cost incurred by the estimation step and the other one is the cost of the assignment step. 
Our analysis then consists of two main parts where each part bounds one of the terms in the RHS above. In particular, we first prove that given the partial comparisons seen so far, we can obtain a relatively good estimation
to the rank of the arrived element, and then in the second part, we show that we can typically find an unassigned position in the close proximity of this estimated rank to assign to it. The following two lemmas capture each part
separately. {In both lemmas, the randomness in the expectation is taken
over the random arrivals and the internal randomness of the algorithm:}

\begin{lemma}[Estimation Cost]\label{lem:estimation}
	In Algorithm~\ref{alg:dense}, $
		\Ex\bracket{\sum_{t=1}^{n}\card{\rank{a_{t}} - \erank{a_{t}}}} = O(n\sqrt{n}).$
\end{lemma}

\begin{lemma}[Assignment Cost]\label{lem:assignment}
	In Algorithm~\ref{alg:dense}, $\Ex\bracket{\sum_{t=1}^{n}\card{\erank{a_{t}} - \pi(a_{t})}} = O(n\sqrt{n}).$
\end{lemma}

Theorem~\ref{thm:dense} then follows immediately from these two lemmas and
Eq~(\ref{eq:cost}). The following two sections are dedicated to the proof of
each of these two lemmas. We emphasize that the main part of the argument is 
the analysis of the assignment cost, i.e., Lemma~\ref{lem:assignment}, and in particular its connection to linear probing. 
The analysis for estimation cost, i.e., Lemma~\ref{lem:estimation}, follows from standard Chernoff bound arguments. 


\subsubsection{Estimation Cost: Proof of Lemma~\ref{lem:estimation}} \label{sec:estimation}

The correctness of the estimation step in our algorithm relies on the following proposition that bounds the probability of the deviation between
the estimated rank and the true rank of each element (depending on the time step it arrives). The proof uses the Chernoff bound for sampling without replacement. 

\begin{restatable}{rPro}{lemquantileestimation}
\label{lem:quantile-estimation}
For any $t > 1$ and any $\alpha \geq 0$,   $\Pr\Paren{\card{\rank{a_{t}} - \erank{a_{t}}} \geq 1 + {\frac{n}{t}} + \alpha \cdot \frac{n-1}{\sqrt{t-1}}} \leq e^{-\Omega(\alpha^2)}.$
\end{restatable}
\begin{proof}
  Fix any $t \in [n]$ and element $a_{t}$ and recall that $\rank{a_{t}}$ denotes the true rank of $a_t$.
  Conditioned on a fixed value for the rank of $a_{t}$, the distribution of the number of elements $r_t$ that arrived
  before $a_t$ and have a smaller rank is equivalent to a sampling without replacement
  process of $t-1$ balls where the urn has $\rank{a_{t}} - 1$ red balls and $n -
  \rank{a_{t}}$ blue balls (and the goal is to count the number of red balls). As such $\Ex[r_t] = \frac{\rank{a_{t}} - 1}{n-1}$
  and by the Chernoff bound for sampling without replacement (Proposition
  \ref{prop:chernoff_negative} with $a = n$ and $b=t-1$), we have:
  \begin{align*}
   \Pr\Paren{\card{r_t - \Ex\bracket{r_t}} \geq \alpha \sqrt{t-1}} \leq 2\cdot
  \exp \Paren{-\frac{2(\alpha \sqrt{t-1})^2}{t-1} } = e^{-\Omega(\alpha^2)}.
  \end{align*}
  
  We now argue that 
  	\begin{align*}
    \Pr\Paren{\card{\rank{a_{t}} - \erank{a_{t}}} \geq 1 + {\frac{n}{t}} + \alpha \cdot \frac{n-1}{\sqrt{t-1}}} \leq \Pr\Paren{\card{r_t - \Ex\bracket{r_t}} \geq \alpha \sqrt{t-1}}. 
\end{align*}
which finalizes the proof by the bound in above equation. 

To see this, note that, 
  \begin{align*}
  \alpha \frac{n-1}{\sqrt{t-1}} \geq \card{  \frac{n-1}{t-1} r_t -
  \rank{a_{t}} } \geq \card{ x_t - \rank{a_{t}} } - \frac{n}{t} \geq \card{
    \erank{a_{t}} - \rank{a_{t}} } - 1 - \frac{n}{t}
    \end{align*}
    
  The first inequality follows from substituting the expectation in $\card{r_t
  - \Ex\bracket{r_t}} \geq \alpha \sqrt{t-1}$ and multiplying the whole
  expression by $(n-1)/(t-1)$. 
  The second inequality just follows from the fact that both the variable
  $x_t$ (defined in step $4$ of Algorithm~\ref{alg:main}) and $\frac{n-1}{t-1} r_t $ are in the interval
  $[\frac{n}{t} r_t, \frac{n}{t} (r_t+1)]$. The fact that $x_t$ is in this
  interval comes directly from its definition in the algorithm and the fact that
  $\frac{n-1}{t-1} r_t $ is in the interval is by a simple calculation (see Proposition~\ref{prop:ratios} in Appendix~\ref{sec:appalg}). 
  The last inequality follows from the fact that
  $\erank{a_{t}} = \ceil{x_t}$. 
\end{proof}

We are now ready to prove Lemma~\ref{lem:estimation}. 

\begin{proof}[Proof of Lemma~\ref{lem:estimation}]
	Fix any $t > 1$; we have, 
	\begin{align*}
    \Ex\bracket{\card{\rank{a_{t}} - \erank{a_{t}}}{-1- \frac{n}{t}} } &\leq \int_{\alpha=0}^{\infty} \Pr\Paren{\card{\rank{a_{t}} -
    \erank{a_{t}}} {-1 - \frac{n}{t}} \geq \alpha \cdot \frac{n-1}{\sqrt{t-1}}}
    \cdot \frac{n-1}{\sqrt{t-1}} \cdot d\alpha \\
    & \leq  \frac{n-1}{\sqrt{t-1}} \cdot \int_{\alpha=0}^{\infty} e^{-\Omega(\alpha^2)}
    \cdot d\alpha = O\Paren{\frac{n}{\sqrt{t}}}.  \tag{by Proposition~\ref{lem:quantile-estimation}}
	\end{align*}
	
  Hence, using the trivial bound for $t=1$ and the bound above for $t>1$ we conclude that:
  \begin{align*}
		\Ex\bracket{\sum_{t=1}^{n}\card{\rank{a_{t}} - \erank{a_{t}}}} = \sum_{t=1}^{n}\Ex\bracket{\card{\rank{a_{t}} - \erank{a_{t}}}} 
    = \sum_{t=1}^{n} O\left(\frac{n}{t} + \frac{n}{\sqrt{t}}\right) = O(n\sqrt{n}). \qedhere
	\end{align*} 
\end{proof}

\subsubsection{Assignment Cost: Proof of Lemma~\ref{lem:assignment}}\label{sec:assignment}

For the second part of the analysis, it is useful to think of sampling a random
permutation in the following recursive way: given a random permutation over $t-1$ elements,
it is possible to obtain a random permutation over $t$ elements by inserting the $t$-th
element in a uniformly random position between these $t-1$ elements. Formally, given $\sigma:[t-1]
\rightarrow [t-1]$, if we sample a position $i$ uniformly from $[t]$ and
generate permutation $\sigma':[t] \rightarrow [t]$ such that:
$$\sigma'(t') = \left\{ \begin{aligned} 
  & i & & \text{~if~} t' = t \\
  & \sigma(t') & & \text{~if~} t' < t \text{~and~} \sigma'(t') < i \\
  & \sigma(t') + 1 & & \text{~if~} t' < t \text{~and~} \sigma'(t') > i \\
\end{aligned}\right.$$
then $\sigma'$ will be a random permutation over $t$ elements. It is simple to
see that just by fixing any permutation and computing the probability of it
being generated by this process.

Thinking about sampling the permutation in this way is very convenient for this
analysis since at the $t$-th step of the process, the relative order of the
first $t$ elements is fixed (even though the true ranks can only be determined
in the end). In that spirit, let us also define for a permutation $\sigma:[t]
\rightarrow [t]$ the event $\O_\sigma$ that $\sigma$ is the relative ordering of
the first $t$ elements:
$$\O_\sigma = \{ a_{\sigma(1)} < a_{\sigma(2)} < \hdots < a_{\sigma(t)} \}.$$

The following proposition asserts that the randomness of the arrival and the inner randomness of the algorithm, ensures that the estimated ranks at each time step
 are chosen \emph{uniformly at random} from all possible ranks in $[n]$. 

\begin{proposition}\label{prop:random-assignment}
The values of $\erank{a_{1}}, \hdots, \erank{a_{n}}$ are i.i.d and uniformly
chosen from $[n]$.
\end{proposition}

\begin{proof}
  First let us show that for any fixed permutation $\sigma$ over $t-1$ elements, the
  relative rank $r_t$ defined in the algorithm is uniformly distributed in
  $\{0,\hdots, t-1\}$. In other words:
  $$\Pr[r_t = i \mid \O_\sigma] = \frac{1}{t}, \qquad \forall i \in \{0, \hdots, t-1\}.$$
  Simply observe that there are exactly $t$ permutations over $t$ elements such
  that the permutation induced in the first $t-1$ elements is $\sigma$. Since we
  are sampling a random permutation in this process, each of these permutation are equally likely to happen. Moreover, since each permutation
  corresponds to inserting the $t$-the element in one of the $t$ positions, we obtain the bound. 

  Furthermore, since the probability of each value of $r_t$ does not depend on the induced
  permutation $\sigma$ over the first $t-1$ elements, then $r_t$ is independent
  of $\sigma$. Since all the previous values $r_{t'}$ are completely determined
  by $\sigma$, $r_t$ is independent of all previous $r_{t'}$ for $t' < t$.

  Finally observe that if $r_t$ is random from $\{0, ..., t-1\}$, then $x_t$ is
  sampled at random from $[0,n]$, so $\erank{a_{t}}$ is sampled at
  random from $[n]$. Since for different values of $t \in [n]$, all $r_t$ are independent, all the values of
  $\erank{a_{t}}$ are also independent.
\end{proof}

Now that we established that $\erank{a_{t}}$ are independent and uniform, our
next task is to bound how far from the estimated rank we have to go in the assignment step, before we are able to assign a learned rank to this element. 
This part of our analysis will be similar to the analysis of the linear
probing hashing scheme. If we are forced to let the learned rank of $a_{t}$ be 
far away from $\erank{a_{t}}$, say $\card{ \erank{a_{t}} -
\pi(a_{t})} > k$, then this necessarily means that all positions in the integer interval $[\erank{a_{t}} - k:\erank{a_{t}} + k]$ must have already been assigned as a learned rank of some element. In the following, 
we bound the probability of such event happening for large values of $k$ compared to the current time step $t$.


We say that the integer interval $I = [i : i+s-1]$ of size $s$ is
\emph{popular} at time $t$, iff at least $s$ elements $a_{t'}$ among the $t-1$ elements that appear before the $t$-th element have
estimated rank $\erank{a_{t'}} \in I$. Since by Proposition~\ref{prop:random-assignment} every element has probability
$s/n$  of having estimated rank in $I$ and the estimated ranks are
independent, we can bound the probability that $I$ is popular using
a standard application of Chernoff bound. 

\begin{restatable}{rClaim}{clmlinearprobing}
\label{clm:linear-probing}
Let $\alpha \geq 1$, an interval of size $s \geq 2\alpha \max\Paren{1, \Paren{\frac{t}{n-t}}^2}$ is popular at time $t$ w.p. $e^{-O(\alpha)}$.
\end{restatable}
\begin{proof} The proof follows directly from the Chernoff bound in Proposition~\ref{prop:chernoff-multi}. For $t' \in [t]$, let $X_{t'}$ be the
  event that $\erank{a_{t'}} \in I$ and $X = \sum_{t'=1}^{t} X_{t'}$, then setting
  $\epsilon = \min(1, \frac{n-t}{t})$ we have that:
  \begin{align*}
  \Pr\paren{I~\text{is popular}} &= \Pr\paren{X \geq s} \leq \Pr\paren{X > (1+\eps) \cdot \epsilon \cdot \Ex[X]} \\
 &\leq \exp \Paren{-\frac{\eps^2 \cdot \Ex[X]}{2}} = e^{-O(\alpha)}
  \end{align*}
  as $\Ex[X] =s \cdot t/n$.
\end{proof}

We now use the above claim to bound the deviation between $\erank{a_t}$ and $\pi(a_t)$. The following lemma, although simple, is the key part of the argument.

\begin{lemma}\label{lem:cost} For any $t \leq n$, we have
 $\Ex \card{ \erank{a_{t}} - \pi(a_{t})} = O(\max\Paren{1, \Paren{\frac{t}{n-t}}^2})$.
\end{lemma}

\begin{proof}
  Fix any $\alpha \geq 1$.  We claim that, if the learned rank of $a_{t}$ is a position which has distance at least $k_{\alpha} =  4\alpha \cdot \max\Paren{1,\Paren{\frac{t}{n-t}}^2}$ from its estimated rank, then necessarily there exists an
  interval $I$ of length at least $2k_{\alpha}$ which contains $\erank{a_t}$ and is popular. 
  
  Let us prove the above claim then. Let $I$ be the shortest integer interval $[a:b]$ which contains $\erank{a_t}$ and moreover both positions $a$ and $b$ are not assigned to a learned rank by time $t$ (by this definition, $\pi(a_t)$ would be either $a$ or $b$). 
  For $\card{\erank{a_t} - \pi(a_t)}$ to be at least $k_{\alpha}$, 
  the length of interval $I$ needs to be at least $2k_{\alpha}$. But for $I$ to have length at least $2k_{\alpha}$, we should have at least $2k$ elements from $a_1,\ldots,a_{t-1}$ to have an estimated rank in $I$: this is simply because $a$ and 
  $b$ are not yet assigned a rank by time $t$ and hence any element $a_{t'}$ which has estimated rank outside the interval $I$ is never assigned a learned rank inside $I$ (otherwise the assignment step should pick $a$ or $b$, a contradiction). 
  
  We are now ready to finalize the proof. It is straightforward that in the above argument, it suffices to only consider the integer intervals $[\erank{a_t} - k_{\alpha} : \erank{a_t}+k_{\alpha}]$ parametrized by the choice of $\alpha \geq 1$. 
  By the above argument and Claim~\ref{clm:linear-probing}, for any $\alpha \geq 1$, we have, 
  \begin{align*}
  \Ex\bracket{\card{\erank{a_{t}} - \pi(a_t)}} &\leq \int_{\alpha=0}^{\infty}    \Pr\Paren{\card{ \erank{a_{t}} - \pi(a_t)} >  k_{\alpha}} \cdot k_{\alpha} \cdot d\alpha \\
	&\leq \int_{\alpha=0}^{\infty} \Pr\Paren{\textnormal{Integer interval $[\erank{a_t} - k_{\alpha} : \erank{a_t}+k_{\alpha}]$ is popular}} \cdot k_{\alpha} \cdot d\alpha\\
  &\!\!\!\!\Leq{Claim~\ref{clm:linear-probing}} O(\max\Paren{1,\Paren{\frac{t}{n-t}}^2}) \cdot \int_{\alpha=0}^{\infty} e^{-O(\alpha)} \cdot \alpha \cdot d\alpha \\
  &=  O(\max\Paren{1,\Paren{\frac{t}{n-t}}^2}) . \qedhere
  \end{align*}
\end{proof}

We are now ready to finalize the proof of Lemma~\ref{lem:assignment}. 

\begin{proof}[Proof of Lemma~\ref{lem:assignment}]
	We have, $\Ex\bracket{\sum_{t=1}^{n}\card{\erank{a_{t}} - \pi(a_{t})}} = \sum_{t=1}^{n} \Ex\bracket{\card{\erank{a_{t}} - \pi(a_{t})}}$ by linearity of expectation.	For any $t < n/2$, the maximum term in RHS of Lemma~\ref{lem:cost} is $1$ and hence in this case, we have $\Ex\bracket{\card{\erank{a_{t}} - \pi(a_{t})}} = O(1)$. Thus, the contribution of the first $n/2-1$ terms to the above summation is only $O(n)$. Also, when $t > n-\sqrt{n}$, we can simply write $\Ex\bracket{\card{\erank{a_{t}} - \pi(a_{t})}} \leq n$ which is trivially true
	and hence the total contribution of these $\sqrt{n}$ summands is also $O(n\sqrt{n})$. It remains to bound the total contribution of $t \in [n/2,n-\sqrt{n}]$. By Lemma~\ref{lem:cost}, 
$$		\sum_{t=n/2}^{n-\sqrt{n}} \Ex\bracket{\card{\erank{a_{t}} - \pi(a_{t})}} \leq O(1) \cdot \sum_{t=n/2}^{n-\sqrt{n}} \paren{\frac{t}{n-t}}^2 = O(n\sqrt{n}),  $$
	where the equality is by a simple calculation (see Proposition~\ref{prop:sum} in Appendix~\ref{sec:appalg}). 
\end{proof}

%% file: lower-bound.tex

\renewcommand{\alg}{\ensuremath{\mathcal{A}}}
\newcommand{\algstar}{\ensuremath{\mathcal{A}^*}}

\newcommand{\OR}{\ensuremath{\textnormal{\textsc{Ranking}}}\xspace}

\newcommand{\COR}{\ensuremath{\textnormal{\textsc{Congested-Ranking}}}\xspace}

\newcommand{\cost}[1]{\ensuremath{\textnormal{\textsf{cost}}_{\alg}(#1)}}

\newcommand{\FU}{\ensuremath{\mathcal{U}}}

\newcommand{\event}{\ensuremath{\mathcal{E}}}

\subsection{A Tight Lower Bound}
\label{sec:lower-bound}

We complement our algorithmic result in Theorem~\ref{thm:dense} by showing that
the cost incurred by our algorithm is asymptotically optimal.

\begin{theorem}\label{thm:lower-bound}
Any algorithm for the secretary ranking problem incurs ${\Omega(n\sqrt{n})}$ cost in expectation.
\end{theorem}

To prove Theorem~\ref{thm:lower-bound}, we first show that no deterministic algorithm can achieve better than
$O(n \sqrt{n})$ inversions and then use Yao's minimax principle to extend the lower bound to randomized algorithms (by simply fixing the randomness of the algorithm to obtain a deterministic one with the same performance over the particular distribution of the input).

The main ingredient of our proof of Theorem~\ref{thm:lower-bound} is an anti-concentration bound for
sampling without replacement which we cast as a balls in bins problem. We start by describing this balls in bin problem and prove the anti-concentration
bound in Lemma~\ref{lem:ballsandbin}. Lemma~\ref{lem:anti-concentration} then
connects the problem of online ranking to the balls in bins problem. We
conclude with the proof of Theorem~\ref{thm:lower-bound}. 

To continue, we introduce some asymptotic notation that is helpful for readability. We write $v = \Theta_1(n)$  if variable $v$ is linear in $n$, but also smaller and bounded away from $n$, i.e., $v = cn$ for some  constant $c$ such that  $0 < c  < 1$.

\begin{lemma}
\label{lem:ballsandbin}
Assume there are $n$ balls in a bin, $r$ of which are red and the remaining $n -
  r$ are blue. Suppose $t < \min(r, n-r)$ balls are drawn from the bin uniformly at
  random without replacement, and let $\event_{k,t,r,n}$ be the event that $k$ out of
  those $t$ balls are red. Then, if $r = \Theta_1(n)$ and $t = \Theta_1(n)$, for
  every $k \in \{0, \hdots, t\}$:
  $\Pr\left(\event_{k,t,r,n}\right) =
  O\left(1/\sqrt{n}\right).$
\end{lemma}
\noindent
Our high level approach toward proving Lemma~\ref{lem:ballsandbin} is as follows: 
\begin{enumerate}

\item We first use a counting argument to show that 
$\Pr\left(\event_{k,t,r,n}\right) = {{r}\choose{k}}  {{n-r}\choose{t-k}}/{{n}\choose{t}}.$

\item We then use Stirling's approximation to show 
${{r}\choose{k}} {{n-r}\choose{t-k}}/{{n}\choose{t}}  =O(n^{-1/2})$  for $k = \lfloor \frac{tr}{n} \rfloor$.

\item Finally, with a max. likelihood argument, we show that $\argmax_{k \in [n]} {{r}\choose{k}}  {{n-r}\choose{t-k}}/{{n}\choose{t}}  \approx   \frac{tr}{n}$. 

\end{enumerate}
\noindent
By combining these, we have, $\Pr\left(\event_{k,t,r,n}\right) \leq \max_{k \in [n]} {{r}\choose{k}}  {{n-r}\choose{t-k}}/{{n}\choose{t}} \leq {{r}\choose{k^*}}  {{n-r}\choose{t-k^*}}/{{n}\choose{t}}$ for $k^* \approx \frac{tr}{n}$ (by the third step), which
we bounded by $O(n^{-1/2})$ (in the second step). The actual proof is however rather technical and is postponed to Appendix~\ref{sec:applower-bound}.

The next lemma shows that upon arrival of $a_{t}$, any position has probability
at least $O\paren{1 /\sqrt{n}}$ of being the correct rank for
$a_{t}$, under some mild conditions.
The proof of this lemma uses the previous
anti-concentration bound for sampling without replacement by considering the
elements smaller than $a_{t}$ to be the red balls and the elements larger than
$a_{t}$ to be the blue balls. For $a_{t}$ to have rank $r$ and be the $k$th
element in the ranking so far, the first $t - 1$ elements previously observed
must contain $k-1$ red balls out of the $r-1$ red balls and $t - k$ blue balls out
of the $n - r $ blue balls.

\begin{lemma}\label{lem:anti-concentration}  Fix any permutation $\sigma$ of $[t]$ 
  and let the event $\O_\sigma$ denote the event that $a_{\sigma(1)} <
  a_{\sigma(2)} < \hdots < a_{\sigma(t)}$. If $\sigma(k) = t$, $k = \Theta_1(t)$
  and $t = \Theta_1(n)$ then for any $r$:
	$		\Pr\paren{\rank{a_{t}} = r \mid \O_{\sigma}} = {O}\paren{1/\sqrt{n}}. 
	$\end{lemma}
	\begin{proof}
		Define $\event_k$ as the event that ``$a_{t}$ is the $k$-th smallest element in 
    $a_{1},\ldots,a_{t}$''. We first have, 
$			\Pr\paren{\rank{a_{t}} = r \mid \O_{\sigma}} = 
      \Pr\paren{\rank{a_{t}} = r \mid \event_k}. $    This is simply because $\rank{a_{t}}$ is only a function of the pairwise comparisons of $a_{t}$ with other elements and does not depend on the ordering of the remaining elements between themselves.
		Moreover, 	
		\begin{align*}
			\Pr\paren{\rank{a_{t}} = r \mid \event_k} &= \Pr\paren{\event_k  \mid \rank{a_{t}} = r} \cdot \frac{\Pr\paren{\rank{a_{t}} = r}}{\Pr\paren{\event_k}} 
      = \Pr\paren{\event_k \mid \rank{a_{t}} = r} \cdot \frac{t}{n}
    \end{align*}
		since $a_{t}$ is randomly partitioned across the $[n]$ elements.

Notice now that conditioned on $\rank{a_{t}} = r $, the event $\event_k$ is
exactly the event $\event_{k-1,t-1,r-1,n-1}$ in the sampling without replacement process
defined in Lemma \ref{lem:ballsandbin}. The $n-1$ balls are all the elements but $a_t$, the
$r-1$ red balls correspond to elements smaller than $a_{t}$, the $n-r$
blue balls to elements larger than $a_{t}$, and $t-1$ balls drawn are the elements arrived before $a_t$.

Finally, observe that $\Pr\paren{r < k \vert \event_k} = 0$, so for $r < k$,
the bound holds trivially. In the remaining cases, $r = \Theta_1(n)$ and we use the bound in Lemma \ref{lem:ballsandbin} with  $t/n = \Theta(1)$ to get the statement.
\end{proof}

Using the previous lemma, we can  lower bound the cost due to the $t$-th
element.  Fix any deterministic algorithm $\alg$ for the online ranking problem.
Recall that $\pi(a_{t})$ denotes the learned rank of the item $a_{t}$
arriving in the $t$-th time step.  For any time step $t \in [n]$, we use
$\cost{t}$ to denote the cost incurred by the algorithm $\alg$ in positioning
the item $a_{t}$. More formally, if $\rank{a_{t}} = i$, we have 
$\cost{t} := \card{i - \pi(a_{(t)})}$. The
following lemma is the key part of the proof. 

\begin{lemma}\label{lem:lower-cost-t}
  Fix any deterministic algorithm $\alg$. For any $t =\Theta_1(n)$, $\Ex\bracket{\cost{t}} = \Omega\paren{\sqrt{n}}$. 
\end{lemma}

\begin{proof}
  Let $\sigma$ be a permutation of $[t]$ and $\O_\sigma$ the event that
  $a_{\sigma(1)} < a_{\sigma(2)} < \hdots < a_{\sigma(t)}$. For any
  deterministic algorithm $\A$, the choice of the position $\pi(a_t)$ where to
  place the $t$-th element depends only on $\sigma$. Let $k = \sigma^{-1}(t)$ be
  the relative rank of the $t$-th element. Since the  distribution of $k$ is
  uniform in $[t]$ (see the proof of Proposition \ref{prop:random-assignment}),
  then we have that:
  $$\Pr\left[\frac{t}{4} \leq k \leq \frac{3t}{4}\right] =
  \frac{1}{2}$$
  Conditioned on that event $k = \Theta_1(t)$ so we are in the conditions of Lemma
  \ref{lem:anti-concentration}. Therefore, the probability of each rank given
  the observations is at most $O(1/\sqrt{n})$. Therefore, there is a constant
  $c$ such that:
  $$\Pr\left[ \card{\rank{a_{(t)}} - \pi(a_{(t)}) } < c \sqrt{n} \text{ }
  \middle\vert \text{ } 
   \frac{t}{4} \leq k \leq \frac{3t}{4} \right] \leq \frac{1}{2}$$

   Finally, we observe that:
  \begin{align*}
    \Ex\bracket{\cost{t}} & \geq \frac{1}{2}\cdot \Ex\left[\card{\rank{a_{(t)}} -
  \pi(a_{(t)})}  \text{ } \middle\vert \text{ }  \frac{t}{4} \leq k \leq
    \frac{3t}{4} \right]  \\ & \geq \frac{1}{2} \cdot c \sqrt{n} \cdot \Pr\left[
    \card{\rank{a_{(t)}} - \pi(a_{(t)})} \geq c \sqrt{n}  
    \text{ } \middle\vert \text{ }   \frac{t}{4}
    \leq k \leq \frac{3t}{4} \right]  \geq \frac{c \sqrt{n}}{4}.
  \end{align*}

	\end{proof}

We are now ready to prove Theorem~\ref{thm:lower-bound}.

\begin{proof}[Proof of Theorem~\ref{thm:lower-bound}]
For any deterministic algorithm, sum the bound in Lemma \ref{lem:lower-cost-t}
for $\Theta(n)$ time steps. For randomized algorithms, the same bound extends via
  Yao's minimax principle. The reason is that
a randomized algorithm can be seen as a distribution on deterministic
algorithms parametrized by the random bits it uses. If a randomized
algorithm obtains less than $O(n \sqrt{n})$ inversions in expectation,
then it should be possible to fix the random bits and obtain a deterministic
algorithm with the same performance.
\end{proof}

%% file: sparse.tex
\section{Sparse Secretary Ranking}
\label{sec:sparse}

Now we consider the special case where the number of positions is very large,
which we call sparse secretary ranking. In the extreme when $m \geq 2^{n+1}-1$ it
is possible to assign a position to each secretary without ever incurring a
mistake. To do that, build a complete binary tree of height $n$ and associate
each position in $[m]$ with a node (both internal and leaf) of the binary
tree such that the order of the positions corresponds to the pre-order
induced by the binary tree
(see figure \ref{fig:binary_tree}). Once the elements arrive in online fashion, insert them in
the binary tree and allocate them in the corresponding position.

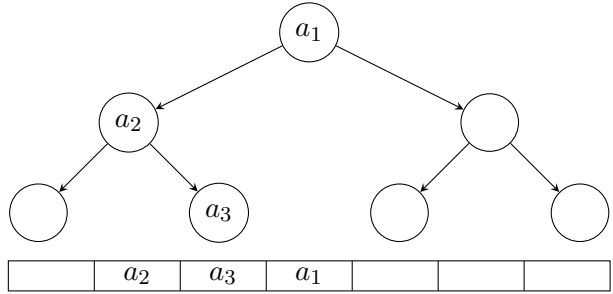
\begin{figure}[h]
\centering
\begin{tikzpicture}[
level 1/.style={sibling distance=60mm},
level 2/.style={sibling distance=30mm},
level 3/.style={sibling distance=15mm},
level 4/.style={sibling distance=7.5mm}, scale=.8
]
\node[c] {$a_1$}
    child{ node[c]  {$a_2$} edge from parent[onearrow]
            child{ node[c] {} }
            child{ node [c] {$a_3$} edge from parent[onearrow] }
    }
    child{ node[c] {} edge from parent[onearrow]
            child{ node [c] {} edge from parent[onearrow] }
            child{ node[c] {} edge from parent[onearrow]}
    };
\begin{scope}[shift={(0,-.3)}]
\draw (-5,-3.5) -- (5,-3.5) -- (5,-4) -- (-5,-4) -- cycle;

\draw (-5+10/7.0,-3.5) -- (-5+10/7.0,-4);
\draw (-5+2*10/7.0,-3.5) -- (-5+2*10/7.0,-4);
\draw (-5+3*10/7.0,-3.5) -- (-5+3*10/7.0,-4);
\draw (-5+4*10/7.0,-3.5) -- (-5+4*10/7.0,-4);
\draw (-5+5*10/7.0,-3.5) -- (-5+5*10/7.0,-4);
\draw (-5+6*10/7.0,-3.5) -- (-5+6*10/7.0,-4);
\node at (-5+3.5*10/7.0,-3.75) {$a_1$};
\node at (-5+1.5*10/7.0,-3.75) {$a_2$};
\node at (-5+2.5*10/7.0,-3.75) {$a_3$};
\end{scope}
\end{tikzpicture}
\label{fig:binary_tree}
\caption{Illustration of the binary tree algorithm for $m = 7$ and order $a_2 <
a_3 < a_1$.}
\end{figure}

We note that the above algorithm works for any order of arrival. If the elements
arrive in random order, it is possible to obtain zero inversions with high
probability for an exponentially smaller value of $m$. The idea is very similar
to the one oulined above. Let $H_n$ be a random variable corresponding to the
height of a binary tree built from $n$ elements in random order. Reed
\cite{reed2003height} shows that $\Ex[H_n] =
\alpha \ln(n)$, $\text{Var}[H_n] = O(1)$ where $\alpha$ is
the solution of the equation $\alpha \ln(2e / \alpha) = 1$ which is
$\alpha \approx 4.31107$.

Since the arrival order of secretaries is uniformly random, the binary tree
algorithm won't touch any node with height more than $\bar
h = \lceil  (\alpha + O(\epsilon)) \ln(n) \rceil$ with probability $1-o(1)$.
This observation allows us to define an algorithm that obtains zero
inversions with probability $1-o(1)$. If $m \geq 2^{\bar h + 1}-1 =
\Omega(n^{2.998 + \epsilon})$, we can build a binary tree with height $\bar h$
and associate each node of the
tree to a position. Once the elements arrive, allocate the item in the
corresponding position. If an item is added to the tree with height larger than
$\bar h$, start allocating the items arbitrarly.

\begin{theorem}\label{thm:sparse}
If $m \geq n^{2.988 + \epsilon}$ then the algorithm that allocates
  according to a binary tree incurs zero inversions with probability $1-o(1)$.
\end{theorem}

Devroye \cite{devroye1986note} bounds the tail of the distribution of $H_n$ as
follows:
$$\Pr[H_n \geq k \cdot \ln n] \leq \frac{1}{n} \cdot \left( \frac{2e}{k}
\right)^{k\cdot  \ln n}$$
for $k > 2$. In particular: $\Pr[H_n \geq 6.3619 \cdot \ln n] \leq 1/n^2$.
Adapting the analysis above, we can show that for $m \geq 4.41$ (where
$4.41 = 6.3619 \cdot \ln(2)$) the algorithm incurs less than one inversion
in expectation.

\begin{corollary}
  If $m \geq \Omega(n^{4.41})$ then the algorithm that allocates according to a
  binary tree incurs $O(1)$ inversion in expectation.
\end{corollary}

\section{General Secretary Ranking}

In the general case, we combine the ideas for the sparse and dense case to
obtain an algorithm interpolating both cases. As described in
Algorithm~\ref{alg:main}, we construct a complete binary search tree of height
$h$ and associating one position
for each internal node, but for the leaves we associate a block of $w = m/2^h -
1$ positions (see Figure \ref{fig:binary_tree_2}). If we insert an element
in a leaf, we allocate according to an instance of the dense ranking algorithm.
By that we mean that the agorithm pretends that the elements allocated to that
leaf are an isolated instance of dense ranking with $w$ elements and $w$
positions. We will set $h$ such that in expectation there only $w$ elements in
each leaf with high probability. If at some point more than $w$ elements are
placed in any given leaf, the algorithm starts allocating arbitrarly.

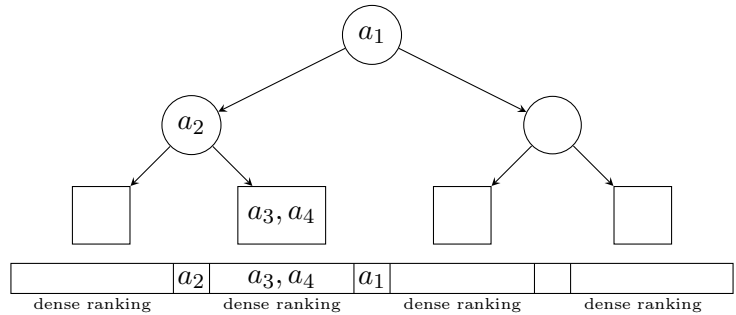
\begin{figure}[h]
\centering
\begin{tikzpicture}[
level 1/.style={sibling distance=60mm},
level 2/.style={sibling distance=30mm},
level 3/.style={sibling distance=15mm},
level 4/.style={sibling distance=7.5mm}, scale=.8
]
\node[c] {$a_1$}
    child{ node[c]  {$a_2$} edge from parent[onearrow]
            child{ node[r] {} }
            child{ node [r] {$a_3,a_4$} edge from parent[onearrow] }
    }
    child{ node[c] {} edge from parent[onearrow]
            child{ node [r] {} edge from parent[onearrow] }
            child{ node[r] {} edge from parent[onearrow]}
    };
\begin{scope}[shift={(0,-.3)}]
\draw (-6,-3.5) -- (6,-3.5) -- (6,-4) -- (-6,-4) -- cycle;

\draw (-6+12/4.0-.3,-3.5) -- (-6+12/4-.3,-4);
\draw (-6+2*12/4-.3,-3.5) -- (-6+2*12/4-.3,-4);
\draw (-6+3*12/4-.3,-3.5) -- (-6+3*12/4-.3,-4);
\draw (-6+12/4.0+.3,-3.5) -- (-6+12/4+.3,-4);
\draw (-6+2*12/4+.3,-3.5) -- (-6+2*12/4+.3,-4);
\draw (-6+3*12/4+.3,-3.5) -- (-6+3*12/4+.3,-4);

\node at (-6+2.7/2.0,-4.2) {\tiny{dense ranking}};
\node at (-6+3+.15+2.7/2.0,-4.2) {\tiny{dense ranking}};
\node at (-6+6+.15+2.7/2.0,-4.2) {\tiny{dense ranking}};
\node at (-6+9+.15+2.7/2.0,-4.2) {\tiny{dense ranking}};

\node at (-6+2*12/4,-3.77  ) {$a_1$};
\node at (-6+1*12/4,-3.77  ) {$a_2$};
\node at (-6+1.5*12/4,-3.77  ) {$a_3,a_4$};

\end{scope}
\end{tikzpicture}
\caption{Illustration of the general algorithm (Algorithm~\ref{alg:main}) for order $a_2 <
a_3 < a_4 < a_1$. The leaves are associated with blocks of $w$ consecutive
positions and internal nodes are associated with a single position. Elements
  $a_3$ and $a_4$ are associated with the same leaf and therefore we place them
  in a block of $w$ positions as we would in a dense ranking problem with $w$
  arrivals and $w$ positions.}
  \label{fig:binary_tree_2}
\end{figure}

\medskip

\begin{algorithm2e}[H]

\caption{General secretary ranking}\label{alg:main}

{ \textbf{Input:} a set of $m$ positions, at most $n$ online arrivals and a height $h$.}

	Construct
a complete binary search tree $T$ of height $h$ and associate one position
for each internal node, and a block of $w = m/2^h -
1$ positions for each leaf such that the order of the positions corresponds to the pre-order
induced by the binary tree

\For{\textnormal{any time step $t \in [n]$ and element $a_{t}$}}{	

Insert $a_t$ in the tree $T$

	\If{\textnormal{$a_t$ reaches an empty internal node}}{

Place $a_t$ in the position corresponding to this internal node
}
\Else{
Place $a_t$ according to an instance of the dense ranking algorithm (Algorithm~\ref{alg:dense}) over the block of positions corresponding to the leaf reached by $a_t$. If there are no position available in that block, place $a_t$ arbitrarily
}
}
\end{algorithm2e}

\medskip

For stating our main theorem and its proof, it is convenient to define the
functions:
$$f(\alpha) = \frac{\alpha \ln(2) - 1}{1 - 2 \alpha \ln(2e / \alpha)} \qquad
g(\alpha) = \frac{1}{1 - 2 \alpha \ln(2e / \alpha)}$$
defined in the interval $(\alpha_0, \infty)$ where $\alpha_0 \approx 4.910$ is
the solution to the equation $1- 2 \alpha_0 \ln(2e / \alpha_0) = 0$. Both 
functions are monotone decreasing from $+\infty$ (when $\alpha = \alpha_0)$ to
zero (when $\alpha \rightarrow \infty$).
We are now ready to state our main theorem:

\begin{theorem}
  Assume $m \geq 10 n \log n$ and let $\alpha \in (\alpha_0, \infty)$ be
  the solution to $\frac{m}{9 n \log n} =
  n^{f(\alpha)}$, then the expected number of inversions of the general
  secretary ranking algorithm with $h = \alpha \ln(n^{g(\alpha)})$ is $\tilde{O}(n^{1.5 - 0.5 g(\alpha)})$.
\end{theorem}

We note that the algorithm smoothly interpolated between the two cases
previously analyzed. When $m = n \log(n)$ then $\alpha \rightarrow \infty$, so
$g(\alpha) \rightarrow 0$ and the bound on the theorem becomes
$\tilde{O}(n^{1.5})$. In the other extreme, when $m \rightarrow \infty$, then
$\alpha \rightarrow \alpha_0$ and therefore $g(\alpha) \rightarrow \infty$, so
the bound on the number of inversions becomes $O(1)$.

\medskip

\begin{proof}
Let $H_t$ be the height of the binary tree formed
by the first $t$ elements. By Devroye's bound \cite{devroye1986note}, the
probability that a random binary tree formed by the first $t := n^{g(\alpha)}$ elements has height
more than $h = \alpha \ln(t)$ is $$\Pr[H_t \geq h] \leq \frac{1}{t} 
(2e/\alpha)^{\alpha \ln t} = t^{\alpha \ln(2e/\alpha)-1}.$$ In case this event
happens, we will use the trivial bound of $O(n^2)$ on the number of inversions,
which will contribute $$n^2 t^{\alpha \ln(2e/\alpha)-1} = n^{1.5 - 0.5/(1 - 2
\alpha \ln(2e / \alpha))}= n^{1.5-0.5 g(\alpha)}$$ to the expectation. From this point on, we consider
the remaining event that $H_t < h$.

Next, we condition on  the first $t$ elements that we denote $b_1, \ldots, b_t$
such that $b_1 < \cdots < b_t$. We note that for each remaining  element $a_i$,
$i > t$, we have $b_j < a_i < b_{j+1}$ with probability $1/(t+1)$ for all $j \in
[t]$. Since $b_1, \ldots, b_t$ are all placed in positions corresponding to
internal nodes, each element has at most  probability $1/t$  of hitting any of
the dense-ranking instances.  Thus, each dense ranking instance receives at most
$n/t$ elements in expectation, and  by a standard application of the Chernoff
bound, the probability that a dense ranking instance sees more than $
9 (n/t) \log n$ elements is $ n^{-3}$. If this is the case
for some dense ranking instance, we again use the $n^2$ trivial bound, which
contributes at most $1$ to the expected number of inversions. For the remainder
of the proof, we assume that each dense ranking instance gets at most $
9(n/t)\log n$ elements.

Next, note that the size of each block is $$w  = \frac{m}{2^{h}} -1 =
\frac{m}{t^{\alpha \ln(2)}} -1  \geq 9 (n/t) \log n$$
where the last equality is by definition of $t$. Thus, no more than $w$ elements are inserted in any leaf.

Let $v_i$ is the number of elements in each of the dense rank instances. 
We note that within the elements in each dense ranking block the arrival order
is random, so we can apply the bound from Section~\ref{sec:alg} and obtain by 
Theorem~\ref{thm:dense} that the
total expected cost from the inversions caused by dense rank is at most $$\sum_i
O\left(v_i^{1.5}\right) \leq \tilde{O}(t \cdot (n/t)^{1.5}) = 
\tilde{O}(n^2 t^{\alpha \ln(2e/\alpha)-1}) = \tilde{O}(n^{1.5-0.5 g(\alpha)}) 
$$ since $\sum_i v_i = n$ and $v_i
\leq (n/t) \log(n)$. 
By the construction there are no inversions between elements
inserted in different leaves and between an element inserted in an internal node
and any other element. Summing the expected number of mistakes from the events $H_t \geq h$ and $H_t
< h$,  we get the bound in the statement of the theorem.
\end{proof}

%% file: prelim.tex

\section{Useful Concentration of Measure Inequalities}
\label{sec:appconcentration}

We use the following two standard versions of Chernoff bound (see,
e.g.,~\cite{ConcentrationBook}) throughout.

\begin{proposition}[Multiplicative Chernoff bound] \label{prop:chernoff-multi}
Let $X_1,\ldots,X_n$ be $n$ independent random variables taking values in $[0,1]$ and let $X:= \sum_{i=1}^{n} X_i$. Then, for any $\eps \in (0,1]$,
\begin{align*}
	\Pr\paren{X \geq (1+\eps) \cdot \Ex\bracket{X}} \leq \exp\paren{-{2\eps^2 \cdot \Ex\bracket{X}}}.
\end{align*}
\end{proposition}

\begin{proposition}[Additive Chernoff bound] \label{prop:chernoff}
Let $X_1,\ldots,X_n$ be $n$ independent random variables taking values in $[0,1]$ and let $X:= \sum_{i=1}^{n} X_i$. Then, 
\begin{align*}
	\Pr\paren{\card{X - \Ex\bracket{X}} > t} \leq 2 \cdot \exp\paren{-\frac{2t^2}{n}}.
\end{align*}
  Moreover, if $X_1,\ldots,X_n$ are \emph{negatively correlated} (i.e. $\Pr[X_i
  = 1, \forall i \in S] \leq \prod_{i \in S} \Pr[X_i = 1]$ for all $S \subseteq
  [n]$), then the upper tail holds:
$\Pr\paren{X - \Ex\bracket{X} > t} \leq \exp\paren{-\frac{2t^2}{n}}$.
\end{proposition}

Moreover, in the above setting, if $X$ comes from a sampling \emph{with} replacement process, then
the inequality holds for both upper and lower tails. For sampling without replacement, we refer
to Serfling \cite{serfling1974probability} for a complete discussion and for Chernoff bounds for negatively correlated random variables see
\cite{panconesi1997randomized}.

\begin{proposition}[Chernoff bound for sampling without
  replacement]\label{prop:chernoff_negative}
  Consider an urn with $a \geq b$ red and blue balls. 
  Draw $b$ balls uniformly from the urn without replacement and let $X$ be the number of red
  balls drawn, then the two sided bound holds:
  $\Pr\paren{\card{X - \Ex\bracket{X}} > t} \leq 2 \cdot
  \exp\paren{-\frac{2t^2}{b}}$.
\end{proposition}

\begin{proof}
  If $X_i$ is the event that the $i$-th ball is red, then since $X_i$ are
  negatively correlated, the upper tail Chernoff bound of $X = \sum_i X_i$ 
  holds. Now, let $Y_i = 1 - X_i$ be the probability that the the $i$-th ball
  is blue and $Y = \sum_i Y_i$. The upper tail for $Y$
  correspond to the lower tail for $X$, i.e.:
  $\Pr\paren{X - \Ex\bracket{X} < t} = \Pr\paren{Y - \Ex\bracket{Y} > t} \leq
  \exp\paren{-\frac{2t^2}{b}}$.
\end{proof}

%% file: appendixAlg.tex
\section{Missing Details from Section~\ref{sec:alg}}
\label{sec:appalg}

The following two straightforward equations are used in the proof of Proposition~\ref{lem:quantile-estimation} and Lemma~\ref{lem:assignment}, respectively. 
For completeness, we provide their proofs here. 

\begin{restatable}{rPro}{propratios}
\label{prop:ratios}
If $1 < t \leq n$ and $0 \leq r \leq t-1$, then $r \cdot \paren{\frac{n}{t}} \leq r \cdot \paren{\frac{n-1}{t-1}} \leq (r+1) \cdot \paren{\frac{n}{t}}$.
\end{restatable}
\begin{proof} 
  $0 \leq r \left( \frac{n-1}{t-1} - \frac{n}{t} \right) = r \frac{n-t}{t(t-1)}
  \leq (t-1) \frac{n-t}{t(t-1)} \leq \frac{n}{t}. $
\end{proof}



\begin{restatable}{rPro}{propsum} For any integer $n > 0$, 
	\label{prop:sum}
		$\sum_{t=1}^{n-\sqrt{n}} \paren{\frac{t}{n-t}}^2 = O(n\sqrt{n})$. 
	\end{restatable}
	\begin{proof}
		By defining $k = n-t$, we have, 
		\begin{align*}
			 \sum_{t=1}^{n-\sqrt{n}} \paren{\frac{t}{n-t}}^2 =  \sum_{k=\sqrt{n}}^{n-1} \paren{\frac{n-k}{k}}^2 \leq \sum_{k=\sqrt{n}}^{n-1} \paren{\frac{n}{k}}^2
		\end{align*}
		
		For $i \in [\sqrt{n}]$, define $K_i := \set{k \mid i \cdot \sqrt{n} \leq k < (i+1) \cdot \sqrt{n}}$. For any $k \in K_i$, we have, $\frac{n}{k} \leq \frac{\sqrt{n}}{i}$. As such, we can write, 
		\begin{align*}
		\sum_{k=\sqrt{n}}^{n-1} \paren{\frac{n}{k}}^2 &= \sum_{i=1}^{\sqrt{n}} \sum_{k \in K_i} \paren{\frac{n}{k}}^2 \leq \sum_{i=1}^{\sqrt{n}} \sum_{k \in K_i} \paren{\frac{\sqrt{n}}{i}}^2 \\
		&\leq \sum_{i=1}^{\sqrt{n}} n \cdot \card{K_i} \cdot \frac{1}{i^2} \leq n\sqrt{n} \cdot \sum_{i=1}^{\sqrt{n}} \frac{1}{i^2} = O(n\sqrt{n}) 
		\end{align*}
		as the series $\sum_{i} \frac{1}{i^2}$ is a converging series. 
	\end{proof}

%% file: appendix.tex
\section{Anti-Concentration for Sampling Without Replacement}
\label{sec:applower-bound}

We prove Lemma~\ref{lem:ballsandbin} restated here for convenience. 

\begin{lemma*}
[Restatement of Lemma~\ref{lem:ballsandbin}]
Assume there are $n$ balls in a bin, $r$ of which are red and the remaining $n -
  r$ are blue. Suppose $t < \min(r, n-r)$ balls are drawn from the bin uniformly at
  random without replacement, and let $\event_{k,t,r,n}$ be the event that $k$ out of
  those $t$ balls are red. Then, if $r = \Theta_1(n)$ and $t = \Theta_1(n)$, for
  every $k \in \{0, \hdots, t\}$:
  $\Pr\left(\event_{k,t,r,n}\right) =
  O\left(1/\sqrt{n}\right).$
\end{lemma*}

To prove Lemma~\ref{lem:ballsandbin}, we will describe the sampling without
replacement process explicitly and bound the relevant probabilities. 
\begin{proposition}
\label{prop:stirling}
  Let $0 < c < 1$ be a constant. Then:
  $$\binom{n}{cn} = \Theta( n^{-1/2} c^{- (cn + 1/2)} (1-c)^{- ((1-c)n +1/2)})$$
\end{proposition}

The notation $y = \Theta(x)$ in the lemma statement means that there are
universal constants $0 < \underline \alpha < \overline \alpha$ independent of
$c$ and $n$ such that $\underline \alpha
  \cdot x \leq y \leq \overline \alpha \cdot x$. The proof is based on the
  following version of Stirling's approximation:
 $ {\sqrt {2\pi }}\ n^{n+{\frac {1}{2}}}e^{-n}\leq n!\leq e\ n^{n+{\frac
 {1}{2}}}e^{-n}.$ which can be written in our notation as: $n! = \Theta(
 n^{n+{\frac {1}{2}}}e^{-n})$. The proof of the previous lemmas follows from
 just expanding the factorials in the definition of the binomial:

\begin{proof}
 Observe that
\begin{align*}
\binom{n}{cn} & = \frac{n!}{(cn)! ((1-c)n)!} 
  = \Theta \left( \frac{n^{n+{\frac {1}{2}}}e^{-n}}{(cn)^{cn+{\frac
  {1}{2}}}e^{-cn}((1-c)n)^{(1-c)n+{\frac {1}{2}}}e^{-((1-c)n)}} \right)
 \end{align*}
  The statement follows from simplifying the right hand side.
\end{proof}

\begin{lemma}
\label{lem:ktrn}
  Assume that $r = \Theta_1(n)$ and $t = \Theta_1(n)$ and $t \leq \min(r, n-r)$,
  then for $k = \lfloor rt / n \rfloor$, we have
$$\frac{{{r}\choose{k}} \cdot {{n-r}\choose{t-k}}}{{{n}\choose{t}}}  =
\mathcal{O}\left(1/ \sqrt{n}\right).$$ \end{lemma}
\begin{proof}
  Start by writing $r = c_r \cdot n$ and $t = c_t \cdot n$ for $0 < c_r, c_t < 1$. It will be convenient to assume that $k = rt / n$ is an integer
  (if not and we need to apply floors, the exact same proof work by keeping
  track of the errors introduced by floor). Then we can write:
First, note that
$$ \frac{{{r}\choose{k}} \cdot
  {{n-r}\choose{t-k}}}{{{n}\choose{t}}} =   \frac{{{c_r n}\choose{c_t c_r n}} \cdot {{(1-c_r)n}\choose{(1-c_r)c_t n}}}{{{n}\choose{c_t n}}}$$
  We can now apply the approximation in Proposition \ref{prop:stirling}
  obtaining:
  $$\Theta \left(  \frac{n^{1/2}c_t^{c_t n+{\frac {1}{2}}}(1-c_t)^{(1-c_t)n+{\frac
  {1}{2}}}}{(c_r n)^{1/2}c_t^{c_t(c_r n)+{\frac {1}{2}}}(1-c_t)^{(1-c_t)(c_r
  n)+{\frac {1}{2}}}((1-c_r)n)^{1/2}c_t^{c_t ((1-c_r)n)+{\frac
  {1}{2}}}(1-c_t)^{(1-c_t)((1-c_r)n)+{\frac {1}{2}}}} \right)$$
  Simplifying this expressoin, we get:
  $\Theta\left(  \left(n c_t  (1-c_t) c_r (1-c_r) \right)^{ -1/2}  \right) =
  \Theta_1\left(1 / \sqrt{n} \right)$.
\end{proof}

\begin{lemma}
\label{lem:max}
Fix any $r,t,n$ such that $r, t \leq n$. Then,
$$\argmax_{k \in [n]} \frac{{{r}\choose{k}} \cdot
  {{n-r}\choose{t-k}}}{{{n}\choose{t}}}  = \left\lfloor t \cdot \frac{r}{n}
  \right\rfloor \text{~or~} \left\lceil t \cdot \frac{r}{n} \right\rceil.$$
\end{lemma}
\begin{proof}
The proof is again simpler if we assume $k = tr/n$ is an integer. If not, the
  same argument works controlling the errors. In that case, let $k_1 = t r /n +
  i$ and $k_2 =  t r /n + i + 1$ and as before, let  $r = c_r n$ and $t =c_t
  n$. Note that
\begin{align*}
\frac{\frac{{{r}\choose{k_1}} \cdot
  {{n-r}\choose{t-k_1}}}{{{n}\choose{t}}} }{\frac{{{r}\choose{k_2}}
  \cdot {{n-r}\choose{t-k_2}}}{{{n}\choose{t}}} } & =
  \frac{{{r}\choose{k_1}} \cdot {{n-r}\choose{t-k_1}}}{{{r}\choose{k_2}}
  \cdot {{n-r}\choose{t-k_2}}}  = \frac{{{c_r n}\choose{c_tc_rn+i}} \cdot
  {{(1-c_r)n}\choose{(1-c_r)c_tn-i}}}{{{c_r n}\choose{c_tc_rn+i+1}} \cdot
  {{(1-c_r)n}\choose{(1-c_r)c_tn-i - 1}}}  = \frac{(c_tc_rn +
  i + 1)\cdot((1-c_t)(1-c_r)n +i + 1)}{((1-c_t)c_rn - i)\cdot ((1-c_r)c_tn-i)}
\end{align*} 

If $i \geq 0$, then the last term is at least $
  \frac{(c_tc_rn)\cdot((1-c_t)(1-c_r)n)}{((1-c_t)c_rn)\cdot ((1-c_r)c_tn)}$
  which is greater than one. If $ i \leq - 1$, then the last term is
  $\frac{(c_tc_rn)\cdot((1-c_t)(1-c_r)n)}{((1-c_t)c_rn)\cdot ((1-c_r)c_tn)} $
  which is smaller than one.
 
  Thus, $\frac{{{r}\choose{k}} \cdot {{n-r}\choose{t-k}}}{{{n}\choose{t}}}$ is increasing as $k$ increases up to $tr / n$ and then decreases. Thus, the maximum is reached at $tr / n$.
\end{proof}

\begin{proof}[Proof of Lemma \ref{lem:ballsandbin}]
 We first use a simple counting argument to obtain an expression for $\Pr\left(\event_{k,t,r,n}\right) $ as a ratio of binomial coefficients. We note that there are ${{r}\choose{k}}$ collections of $k$ red balls, ${{n-r}\choose{t-k}}$ collections of $t - k$ blue balls, and that the total number of collections of $t$ balls is  ${{n}\choose{t}}$. Since the $t$ balls are drawn uniformly at random without replacement, we get
$$			\Pr\paren{\event_{k,t,r,n} } = \frac{{{r}\choose{k}} \cdot {{n-r}\choose{t-k}}}{{{n}\choose{t}}}.$$
The $O(1 / \sqrt{n})$ bound now follows directly from Lemma \ref{lem:ktrn} and Lemma \ref{lem:max}.
\end{proof}

\section{Hardness of Online Ranking with Adversarial Ordering}
\label{sec:appadversarial}

\begin{proposition}
\label{prop:adversarial}
If the ordering $\sigma$ of the arrival of elements is adversarial, then any algorithm has cost $\Omega(n^2)$ in expectation.
\end{proposition}
\begin{proof}
At a high level, we construct an ordering such that at each iteration, the arrived element is either the largest or smallest element not yet observed with probability $1/2$ each. Since the algorithm cannot distinguish between the two cases, it suffers a linear cost in expectation at each arrival.

Formally, we define $\sigma$ inductively. At round $t$, let $i_{t,-}$ and $i_{t,+}$ be the minimum and maximum indices of the elements arrived previously. We define $\sigma(t)$ such that  $\sigma(t) = a_{i_{t,-} + 1}$ with probability $1/2$ and $\sigma(t) = a_{i_{t,+} - 1}$ with probability $1/2$. Thus, the $t$th element arrived is either the smallest or largest element not yet arrived.

The main observation is that the pairwise comparisons at time $t$ are identical whether $a_{(t)} =  a_{i_{t,-} + 1}$ or $a_{(t)} = a_{i_{t,+} - 1}$. This is since all the elements previously arrived are either maximal or minimal and there is no elements that are between $a_{i_{t,-} + 1}$ and $a_{i_{t,+} - 1}$ that have previously arrived. Thus the decision of the algorithm is \emph{independent} of the randomization of the adversary for the $t$th element. Thus for any learned rank at time $t$, in expectation over the randomization of the adversary for the element arrived at time $t$, the learned rank is at expected distance  of the true rank at least $n/4$ for $t \leq n/2$. Thus the total cost is $\Omega(n^2)$ in expectation.
 \end{proof}